\documentclass[journal]{IEEEtran}

\newtheorem{definition}{Definition}

\newenvironment{proof}{\begin{IEEEproof}}{\end{IEEEproof}}
\usepackage{amssymb}
\newtheorem{strategy}{Strategy}    
\newtheorem{lemma}{Lemma}  
\newtheorem{property}{Property}
\usepackage{multirow}  
\usepackage{hyperref}

\usepackage{amsmath}
\usepackage{algorithm}
\usepackage{algorithmic}

\usepackage{color}

\ifCLASSINFOpdf
\usepackage[pdftex]{graphicx}
\DeclareGraphicsExtensions{.pdf,.jpeg,.png}
\else
\usepackage[dvips]{graphicx}

\DeclareGraphicsExtensions{.eps}
\fi

\ifCLASSINFOpdf

\else

\fi

\begin{document}
\title{Flexible Pattern Discovery and Analysis}

\author{Chien-Ming Chen,~\IEEEmembership{Senior Member,~IEEE,} 
	Lili Chen, Wensheng Gan*,~\IEEEmembership{Member,~IEEE}

\thanks{This research was partially supported by National Natural Science Foundation of China (Grant No. 62002136), Guangzhou Basic and Applied Basic Research Foundation (Grant No. 202102020277). (Corresponding author: Wensheng Gan)}

	\thanks{Chien-Ming Chen and Lili Chen are with the College of Computer Science and Engineering, Shandong University of Science and Technology, Qingdao 266590, China. (E-mail: chienmingchen@ieee.org and lilichien3@gmail.com)}

	\thanks{Wensheng Gan is with the College of Cyber Security, Jinan University, Guangzhou 510632, China. (E-mail: wsgan001@gmail.com)}
	
}

\maketitle

\begin{abstract}
	Based on the analysis of the proportion of utility in the supporting transactions used in the field of data mining, high utility-occupancy pattern mining (HUOPM) has recently attracted widespread attention. Unlike high-utility pattern mining (HUPM), which involves the enumeration of high-utility (e.g., profitable) patterns, HUOPM aims to find patterns representing a collection of existing transactions. In practical applications, however, not all patterns are used or valuable. For example, a pattern might contain too many items, that is, the pattern might be too specific and therefore lack value for users in real life. To achieve qualified patterns with a flexible length, we constrain the minimum and maximum lengths during the mining process and introduce a novel algorithm for the mining of flexible high utility-occupancy patterns. Our algorithm is referred to as HUOPM$^+$. To ensure the flexibility of the patterns and tighten the upper bound of the utility-occupancy, a strategy called the length upper-bound (LUB) is presented to prune the search space. In addition, a utility-occupancy nested list (UO-nlist) and a frequency-utility-occupancy table (FUO-table) are employed to avoid multiple scans of the database. Evaluation results of the subsequent experiments confirm that the proposed algorithm can effectively control the length of the derived patterns, for both real-world and synthetic datasets. Moreover, it can decrease the execution time and memory consumption.
\end{abstract}


\IEEEpeerreviewmaketitle

\section{Introduction}

\IEEEPARstart{T}{he} initial motivation for frequent pattern mining (FPM) was to analyze the shopping behavior of customers using transactional databases and recommend frequently purchased patterns to customers \cite{agrawal1994fast,han2000mining,chen1996data,mannila1997discovery,goethals2003survey}. In this case, researchers believed that the item is binary and whether an item appears in a transaction is considered primary. However, frequent purchase patterns are occasionally less profitable than infrequent purchase patterns with high profits, which poses a fundamental problem. Hence, the discovery of high-utility patterns that consider not only the internal utility (e.g., quantity) but also the external utility (e.g., profit, interest, or weight) \cite{shen2002objective,chan2003mining,yao2006mining,yao2006unified} has gained substantial research attention. Moreover, a framework called high-utility pattern mining (HUPM) \cite{gan2021survey,ahmed2009efficient} was proposed to address this practical issue. In contrast with frequent pattern mining (FPM), the lack of a downward closure property makes HUPM more difficult and intractable.

Up until now, numerous approaches and strategies have been designed to increase the efficiency and convenience of mining high-utility patterns \cite{li2008isolated,liu2005two}. These methods can be broadly divided into three major categories. The first category involves the candidate generation-and-test approach \cite{yao2006mining,li2008isolated,liu2005two}: The eligible patterns are selected as candidates based on an upper-bound evaluation and then calculated and tested to determine whether they are qualified. The second category entails tree-based algorithms \cite{ahmed2009efficient,ahmed2011huc,tseng2012efficient,tseng2010up}. The necessary information is projected onto a tree in the database, which can avoid multiple traversals of the database. The third category employs vertical data structures (e.g., utility-lists) \cite{lin2016fhn,fournier2014fhm}. Similar to the tree-based approach, crucial information is stored in a vertical data structure, and the utility of any pattern can be calculated using this structure. HUPM has dozens of practical applications in real life, such as gene regulation \cite{zihayat2017mining} and web click-stream analysis \cite{li2008fast,shie2010online}. A challenging problem arises in selecting a pattern to represent the transactions.

To the best of our knowledge, only a few studies have been carried out on this problem. The concept of occupancy \cite{tang2012incorporating} was originally defined as the share of the number of items to one of the transactions in which the item existed. This requires the patterns to occupy the majority of the supporting transactions. In addition, an occupancy measurement can be widely used in pattern recommendation. When a user browses a website, if the number of times the user clicks on a URL is greater than a certain threshold, then the URL has a high occupancy pattern. However, if this user browses for a long time at a URL that is not frequently clicked, the URL may be valuable. Hence, Shen \textit{et al.} \cite{shen2016ocean} incorporated occupancy and utility to create an original concept called the utility-occupancy and designed an algorithm called OCEAN. Although this algorithm is novel, it has a limitation in that it may exclude certain patterns that should be qualified. To address this drawback, Gan \textit{et al.} \cite{gan2019huopm} proposed an efficient algorithm called high utility-occupancy pattern mining. Taking advantage of two compact list structures, namely, a utility-occupancy list and a frequency-utility-occupancy table, respectively, HUOPM can reduce the running time and memory consumption, which are two of its principal merits. It is clear that utility-occupancy has a wide applicability in an information-driven society. For instance, during a holiday, tourists may want to go out but not know where to visit. They can then check the travel route recommendation, which analyzes and calculates the utility-occupancy of the tourist route, and finally recommends the best route for the tourists.

During the process of pattern mining, researchers tend to extract all patterns. However, they may not necessarily be useful to actual production or management. For example, supermarket managers generally display milk and bread together, but they rarely bundle \{milk, bread, strawberry jam, and gum\}. Although both sets are elegant patterns selected by the algorithm, the shorter one is obviously more popular with decision makers. In previous studies, Pei \textit{et al.} \cite{pei2002constrained} applied length constraints on frequent pattern mining by appending no items after the patterns. Next, Fournier-Viger \textit{et al.} \cite{fournier2016fhm} proposed the FHM+ algorithm focusing on utility, which improved the FHM by incorporating the length constraints into utility mining. This narrows the upper bound of the patterns and further trims the search space. Nevertheless, no method has been proposed in the field of utility-occupancy to address the problem of a length constraint.

In this study, we focus on mining flexible high utility-occupancy patterns by developing a novel algorithm called HUOPM$^+$. The proposed algorithm is dedicated to discovering high utility-occupancy patterns with length constraints; in addition, it is a generic framework (or called an extension) of the state-of-the-art HUOPM algorithm. The major contributions of this paper are briefly summarized as follows.

\begin{itemize}
	\item A generic and practical algorithm is proposed to exploit flexible high utility-occupancy patterns. During the execution of this algorithm, the minimum and maximum lengths are needed in advance to determine the length range of the derived patterns.
	
    \item To avoid scanning the database multiple times, two compact data structures, called a utility-occupancy nested list (UO-nlist) and a frequent-utility-occupancy table (FUO-table), are constructed to store vital information from the databases.
    
    \item It is recommended to tighten the upper bound with the newly designed LUB, which is less than the original upper-bound in the HUOPM algorithm.
    
	\item Subsequent experiments have been carried out on both real-world and synthetic datasets, the results show that all patterns with length constraints can be obtained and that the efficiency of the proposed algorithm is significantly high in terms of the execution time and memory consumption.
\end{itemize}

The remainder of this paper is broadly organized as follows. Related studies are introduced in Section \ref{sec:2}. In Section \ref{sec:background}, some fundamental knowledge regarding this study is presented. The presented HUOPM$^+$ algorithm and three novel pruning strategies are detailed in Section \ref{sec:4}. In Section \ref{sec:experiments}, subsequent experiments confirming the effectiveness and efficiency of the proposed algorithm are described. Finally, Section \ref{sec:conclusion} provides some concluding remarks regarding this research.

\section{Related Studies}
\label{sec:2}

The studies related to the HUOPM$^+$ algorithm mainly deal with three aspects, i.e., high-utility pattern mining, high utility-occupancy pattern mining, and flexible pattern mining, which are discussed below.

\subsection{High-Utility Pattern Mining}

Thus far, numerous studies have been carried out on HUPM, which aims at miming qualified patterns whose utilities are greater than or equal to a predefined minimum utility threshold. Because HUPM provides guidance for many applications such as decision-making, it has attracted significant attention. HUPM was initially proposed in \cite{yao2004foundational}. Each pattern has two aspects, one being the quantity of items whose alias is an internal utility, and the other being a unit utility (e.g., profit), which is defined by experts as an external utility. If any patterns satisfy the minimum utility threshold, they can be derived as high-utility patterns. HUPM is technically more challenging than FPM because FPM has a downward closure property, whereas HUPM does not. To overcome this critical issue, Liu \textit{et al.} \cite{liu2005two} creatively proposed a two-phase algorithm and established a transaction-weighted utilization (TWU) model by adopting the anti-monotonic property of a transaction-weighted utility. Liu \textit{et al.} \cite{liu2012mining} then presented HUI-Miner, achieving a better performance than previous approaches. This algorithm scans the transaction database twice to construct the initial utility-list, and thus there is no longer a need to access the database. The utility-list contains the utility and the remaining utility of the patterns, which is a necessary condition for calculating the upper bound of the extended patterns. The above algorithms are all itemset-based utility mining algorithms, although different types of data also exist. Utility mining of temporal data \cite{lin2015efficient}, uncertain data \cite{lin2016efficient}, dynamic data \cite{lin2015incremental,yun2019efficient,nguyen2019mining}, sequence data \cite{gan2020proum,gan2020fast}, and other factors are all interesting research directions, as highlighted in a literature review \cite{gan2021survey}.

\subsection{High Utility-Occupancy Pattern Mining}

In terms of the contribution ratio of a pattern, the above presented utility approaches are useless. Thus, it is sensible to introduce a new definition, i.e., occupancy. The initial concept of occupancy, whereby occupancy is defined as the number of occurrences of a pattern, was designed by Tang \textit{et al.} \cite{tang2012incorporating}. Unfortunately, this method is unsuitable for research on utility mining. Subsequently, Shen \textit{et al.} \cite{shen2016ocean} started conducting research on utility-occupancy and developed a representative algorithm called OCEAN to find patterns whose share of utility in the supporting transactions is greater than a specific value. However, the OCEAN algorithm fails to discover complete high utility-occupancy patterns.  Gan \textit{et al.} \cite{gan2019huopm} proposed a successful and efficient algorithm called HUOPM to address this disadvantage. Two compact data structures are applied to deposit the vital data in the database. A utility-occupancy list (UO-list) is used to store the utility and the remaining utility of each pattern. Each entry records the details of a transaction. An FUO-table is then obtained for the convenience of integrating and revising the information in the UO-list. HUOPM just focuses on precise data, thus Chen \textit{et al.} \cite{chen2021discovering} recently extended HUOPM to deal with the uncertainty in uncertain data. However, these algorithms can not discover flexible patterns based on various constraints.

\subsection{Flexible Pattern Mining}

Through this research, we found that the patterns adopted by managers generally have a shorter length and that longer patterns are not universal because they are too specific. Hence, it is advisable for users to pursue a flexible algorithm. According to the demand, the generation of a large number of patterns considerably decreases the efficiency of the mining. Pei et al. \cite{pei2002constrained} emphasized adding a constraint-based approach during frequent pattern mining. The authors mainly adopted a pattern-growth method, which is applicable to various constraints. When this method is applied to the field of utility mining, however, it is too superficial and does not penetrate into the interior of the utility mining algorithms. FHM$^+$ \cite{fournier2016fhm} has an interesting feature and it discovers high-utility item sets with length constraints. The authors considered the maximum length of the patterns as the dominant control parameter, and designed length upper-bound reduction (LUR) to prune the search space. Several studies regarding flexible sequential pattern mining such as CloSpan \cite{yan2003clospan}, BIDE \cite{wang2004bide}, and MaxFlex \cite{arimura2007mining} have also been conducted to meet the requirements of different applications.


\section{Preliminary and Problem Statement} 
\label{sec:background}

To assist with this discussion, in this section, some symbols in connection with the proposed algorithm are introduced and defined. Provided that there are several items in a transaction database, and let $I$ = \{$i_{1}$, $i_{2}$, $\ldots$, $i_{m}$\} be a collection of m distinct items, if an itemset contains \textit{k} distinct items, it can then be denoted as $k$-itemset. A transaction database $\mathcal{D}$ consists of n transactions in which $\mathcal{D}$ = \{$T_{1}$, $T_{2}$, $\ldots$, $T_{n}$\}. Each transaction holds its own particular identifier \textit{tid} and is a subset of $I$. Every transaction includes three parts, i.e., the transaction identifier \textit{tid}, item name, and number of purchases of the corresponding item. Next, we take advantage of the transaction database in TABLE \ref{table:db1} and the profit presentation of each item in TABLE \ref{table:profit} as a running example.

\begin{table}[!htbp]
	\centering
	\small
	\caption{Transaction database.}
	\label{table:db1}
	\begin{tabular}{|c|c|c|c|}
		\hline
		\textbf{\textit{tid}} & \textbf{Transaction (item, quantity)} & \textbf{\textit{tu}} \\ \hline \hline
		$ T_{1} $ & 	\textit{a}:3, \textit{b}:4, \textit{c}:2, \textit{d}:6, \textit{e}:2  &  \$63 \\ \hline
		$ T_{2} $ & 	\textit{a}:7, \textit{b}:4, \textit{c}:1, \textit{e}:2   & \$62 \\ \hline
		$ T_{3} $ &	    \textit{a}:5, \textit{b}:2, \textit{e}:1  &  \$35 \\ \hline
		$ T_{4} $ &	    \textit{b}:4, \textit{c}:1, \textit{d}:2  &  \$25 \\ \hline
		$ T_{5} $ &	    \textit{a}:2, \textit{d}:4 &   \$14  \\ \hline
		$ T_{6} $ &	     \textit{a}:2, \textit{b}:2, \textit{c}:6, \textit{d}:4, \textit{e}:3  &  \$60 \\ \hline
		$ T_{7} $ &	     \textit{a}:1, \textit{b}:2 &  \$13 \\ \hline
		$ T_{8} $ &	     \textit{d}:3  &  \$6 \\ \hline
		$ T_{9} $ &  	\textit{b}:3, \textit{c}:5, \textit{d}:2, \textit{e}:5 &  \$74 \\ \hline
		$ T_{10} $	&	\textit{b}:3, \textit{e}:5 &  \$65 \\ \hline	
	\end{tabular}
\end{table}

\begin{table}[!htbp]
	\centering
	\small
	\caption{Unit utility of each item}
	\label{table:profit}
	\begin{tabular}{|c|c|c|c|c|c|}
		\hline
		\textbf{Item}    &  $ a $  &  $ b $  &  $ c $   & $ d $  & $ e $  \\ \hline
		\textbf{Utility (\$)} &  $ 3 $  &  $ 5 $  &  $ 1 $  & $ 2 $  & $ 10 $ \\ \hline
	\end{tabular}
\end{table}

\begin{definition}
	\label{def_1}
	\rm The number of itemsets $X$ contained in a database is usually defined as \textit{support count} \cite{agrawal1994fast}, which can be denoted as \textit{SC(X)}. Given a \textit{support threshold} $\alpha$ ($ 0 < \alpha \leq 1 $), if $SC(X) \geq$ $\alpha$  $\times$ $|D|$, we can determine that $X$ is a frequent pattern. The collection of transactions containing $X$ is expressed as $ \varGamma_X$. That is, if an itemset $X$ appears in the transaction $T_q$, then $T_q$ belongs to $ \varGamma_X$, and naturally, $ SC(X)$ = $|\varGamma_X| $.
\end{definition}

As shown in TABLE \ref{table:db1}, the pattern $ac$ appears in transactions $T_1$, $T_2$, and $T_6$. Therefore, $SC (ac) $ = 3 and $\varGamma_{(ac)}$ = $ \{T_1$, $T_2$, $T_6\} $. Assuming that the value of $\alpha$ is 0.3, $SC(ac)$  $\geq $ $\alpha$ $\times$ $|D| $, and thus $ac$ is frequent. 

In a market basket analysis, based on the preference of the customers for the products or an evaluation of experts regarding the items, each item in the database is associated with a positive number, or in other words, a unit profit. 
 
\begin{definition}
	\label{def_4}
	\rm We define the \textit{internal utility} of an item $i$ in transaction $T_q$ as \textit{iu(i, $T_q$)}, which refers to the number of occurrences of the corresponding item. We define the \textit{external utility} of an item $i$ in the database as \textit{eu(i)}, namely, the unit utility of $i$ in TABLE \ref{table:profit}, which is generally set subjectively.
\end{definition}

\begin{definition}
	\label{def_5}
	\rm  The utility of an item $i$ in the supporting transaction $T_{q}$ is defined as $u(i, T_{q})$ = $iu(i, T_{q})$ $\times$ $eu(i)$. Moreover, the utility of a pattern $X$ in a transaction $T_{q}$ is equal to the total utility of each item in the pattern and is represented as $u(X, T_{q})$ = $\sum _{i \in X  \wedge X \subseteq T_{q}} u(i, T_{q}) $. Hence, the utility of $X$ existing in a transaction database $\mathcal{D}$ is denoted as $u(X)$ = $\sum_{X\subseteq T_{q}\wedge T_{q} \in \mathcal{D}} u(X, T_{q}) $. A summary of all utilities in a transaction is recorded as the transaction utility (\textit{tu}).
\end{definition}

Let us take $e$ and $ac$ as examples for an easier understanding of the utility calculation. Here, $ u(e) $ = $ u(e, T_1) $ + $ u(e, T_2) $ + $ u(e, T_3) $ + $ u(e, T_{6})$ + $ u(e, T_{9})  $ + $ u(e, T_{10})$ = \$20 + \$20 + \$10 + \$30 + \$50 + \$50 = \$180, $ u(ac) $ = $ u(ac, T_1) $ + $ u(ac, T_2) $ + $ u(ac, T_6) $  = \$11 + \$22 + \$12 = \$45. Each transaction utility refers to the last column of TABLE \ref{table:db1}.

The \textit{occupancy} was originally proposed to measure the proportion of pattern $X$ in the supporting transactions \cite{tang2012incorporating}. In the field of utility mining, it is crucial to discover HUOPs \cite{shen2016ocean,gan2019huopm}.

\begin{definition}
	\label{def_7}
	\rm  Assume $X$ is present in transaction $T_q$, and the utility-occupancy of $X$ in supporting transaction $T_q$ is defined as follows: 
	\begin{equation}
	uo(X, T_q) = \dfrac{u(X, T_q)}{tu(T_q)}.
	\end{equation}

	Assume $\varGamma_X$ is a collection of all transactions containing pattern $X$. The utility-occupancy of a pattern in a database is calculated as follows:
	\begin{equation}
		uo(X) = \dfrac{\sum_{X \subseteq T_q \wedge T_q \in \mathcal{D}}uo(X,T_q)}{|\varGamma_X|}.
\end{equation}  
\end{definition}

 \begin{definition}
	\label{def_9}
	\rm Let there be a transaction database $\mathcal{D}$. A pattern $X$ is denoted as a HUOP if and only if $SC(X)$ $\geq$ $\alpha$ $\times$ $|\mathcal{D}|$ and $ uo(X)$ $\geq$ $\beta $, where $ \alpha $ ($ 0 < \alpha \leq 1 $) is the predefined minimum support threshold and  $ \beta $ ($ 0 < \beta \leq 1 $) is the predefined minimum utility-occupancy threshold.
\end{definition}

In a previous study, the full transaction utility and support were calculated. It is convenient to obtain $SC(ac)$ = 3, $tu(T_1)$ = \$63, $tu(T_2)$ = \$62, and $tu(T_6)$ = \$60.	Therefore, $ uo(ac, T_1) $ is calculated as \$11/\$63 $\approx $ 0.1746. Similarly, $ uo(ac, T_2) $ and  $ uo(ac, T_6) $ are calculated as 0.3548 and 0.2, respectively. Furthermore, $uo(ac)$ = \{0.1746 + 0.3548 + 0.2\} /3 $\approx $ 0.2431. Provided that the value of $\alpha$ and $\beta$ is 0.3, the pattern $ac$ is not a HUOP.

There have been dozens of studies conducted on HUOPM; however, one outstanding issue of HUOPM algorithms is that they normally focus on discovering massive patterns containing numerous items. As we discussed earlier, these patterns might not work for users because they usually represent an unusual situation. To increase the utility of the discovered patterns, we address the issue of mining flexible high utility-occupancy patterns with a length constraint. The minimum length \textit{minlen} and the maximum length \textit{maxlen} of the required patterns are predefined.

\begin{definition}
	\label{def_13}
	\rm (\textit{Flexibly mining high utility-occupancy patterns}) The flexibly mining of high utility-occupancy patterns aims to discover HUOPs containing up to \textit{maxlen} items and at least \textit{minlen} items.
\end{definition}

\textbf{Problem Statement.} Consider the existence of a given transaction database $\mathcal{D}$, a utility-table recording the unit utility of each item, and four input parameters ($\alpha$, $\beta$, \textit{minlen}, and \textit{maxlen}) used as the mining constraints. The purpose of this study is to mine flexible eligible patterns with the length being at least \textit{minlen} and at most \textit{maxlen}, under the condition that not only is the support count greater than or equal to the minimum support count threshold $\alpha$ $\times$ $|D|$, the value of the utility-occupancy is also no less than the minimum utility-occupancy threshold $\beta$ $\times$ $|\mathcal{D}|$.

Assuming that \textit{minlen} = 1 and \textit{maxlen} = 3, the length of all derived patterns should range from 1 and 3. Thus, although \{$caeb$\} is a HUOP, it is not the desired pattern because its length exceeds \textit{maxlen}. In addition, assuming that the value of $\alpha$ and $\beta$ are both 0.3, all of the desired patterns can be obtained from Table \ref{table:patterns}.

\begin{table}[!htbp]
	\centering
	\small
	\caption{Desired patterns with length constraints}
	\label{table:patterns}
	\begin{tabular}{|c|c|c|c|c|c|}
		\hline
		\textbf{HUOP}    &  \textbf{\textit{sup}}  &  \textbf{\textit{uo}}  &   \textbf{HUOP}    &  \textbf{\textit{sup}}  &  \textbf{\textit{uo}}   \\ \hline \hline
		$d$ &   6   &   0.3515   & $db$ &   4   &   0.5062   \\ \hline
		$e$ &   6   &   0.4784   &  $ eb $  &  6  &  0.7328  \\ \hline
		$b$ &   8   &   0.3869   &  $ cae $  &  3   &  0.6232  \\ \hline
		$ce$ &   4   &   0.5078   &  $ cab $  &  3   &  0.5120  \\ \hline
		$cb$ &   5   &   0.4130   &  $ cde $  &  3   &  0.6901  \\ \hline
		$ad$ &   3   &   0.5222   &  $ cdb $  &  4   &  0.5660  \\ \hline
		$ae$ &   4   &   0.6090   &  $ ceb $  &  4   &  0.7601  \\ \hline
		$ab$ &   5   &   0.6205   &  $ aeb $  &  4   &  0.8821  \\ \hline  
		$de$ &   3   &   0.6237   &  $ deb $  &  3   &  0.8526  \\ \hline
	\end{tabular}
\end{table}

\section{Proposed Flexible HUOPM$^+$ Algorithm} 
\label{sec:4}

In this section, some novel definitions based on two length constraints are first describe. Then, according to these definitions and combined with the utility-list \cite{liu2012mining}, two novel data structures are further constructed. A UO-nlist and a FUO-table are designed to maintain the information regarding the given database. In addition, to avoid an exhaustive search, several pruning strategies are proposed to further narrow the upper bound of utility-occupancy on the subtree nodes in a $SC$-tree, the definition of which is introduced in the next subsection. Finally, the designed algorithm should be described with the help of a pseudocode. The detailed flowchart of the proposed HUOPM$^+$ algorithm is shown in Fig. \ref{fig:flowchart}.

\begin{figure*}[!htbp]
	\centering
	\includegraphics[scale=0.48]{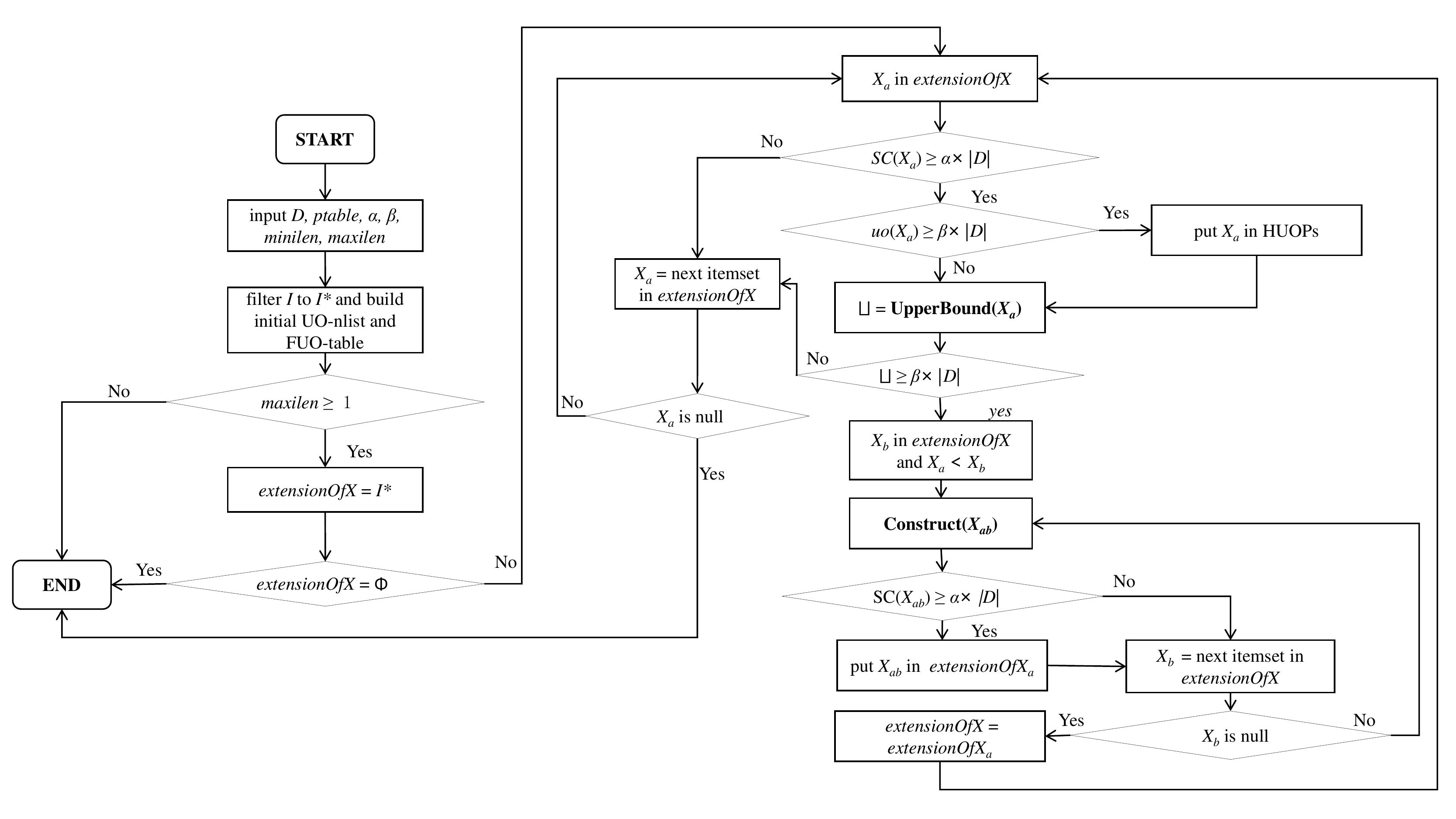}
	\caption{Flowchart of the proposed HUOPM$^+$ algorithm.}
	\label{fig:flowchart}
\end{figure*}

\subsection{Revised Remaining Utility-Occupancy}

As is well-known, the HUOPM algorithm \cite{gan2019huopm} adopts a depth-first search strategy. Thus, an intuitive method for controlling the length of the discovered patterns is to output only those patterns with no less than the minimum length and to stop the extension, where the number of items of a pattern is equal to the established maximum length. Nevertheless, the advantages of this approach are not obvious and are too superficial, which may fail to decrease the number of visited nodes within the specified length of range. As the reason behind this approach, it is unable to reduce the upper bound of the utility-occupancy of the patterns and prune the search range. To handle this drawback, we developed a novel strategy called the LUB, which provides a revision of the remaining utility-occupancy for discovering the HUOPs. A description of this approach is detailed below.


When we run the program, we should traverse the items in each transaction in a certain order. Without a loss of generality, in this study, we take the support of the ascending order as the arrangement and denote it as $\prec$. For example, from the database shown in TABLE \ref{table:db1}, we can easily see that $SC(c)$ $<$ $ SC(a)$ $\leq $ $SC(d)$ $\leq$ $SC(e)$ $<$ $SC(b)$. Therefore, the support-ascending order is $ c \prec a \prec d \prec e \prec b $. TABLE \ref{table:db} shows the revised database modified from TABLE \ref{table:db1} according to the order of $\prec$.

\begin{table}[!htbp]
	\centering
	\small
	\caption{Revised transaction database.}
	\label{table:db}
	\begin{tabular}{|c|c|c|c|}
		\hline
		\textbf{\textit{tid}} & \textbf{Transaction (item, quantity)} & \textbf{\textit{tu}} \\ \hline  \hline
		$ T_{1} $ & 	\textit{c}:2, \textit{a}:3, \textit{d}:6, \textit{e}:2, \textit{b}:4    &  \$63 \\ \hline
		$ T_{2} $ & 	\textit{c}:1, \textit{a}:7, \textit{e}:2, \textit{b}:4     & \$62 \\ \hline
		$ T_{3} $ &	    \textit{a}:5, \textit{e}:1, \textit{b}:2   &  \$35 \\ \hline
		$ T_{4} $ &	    \textit{c}:1, \textit{d}:2, \textit{b}:4    &  \$25 \\ \hline
		$ T_{5} $ &	    \textit{a}:2, \textit{d}:4 &   \$14  \\ \hline
		$ T_{6} $ &	    \textit{c}:6, \textit{a}:2, \textit{d}:4, \textit{e}:3, \textit{b}:2   &  \$60 \\ \hline
		$ T_{7} $ &	     \textit{a}:1, \textit{b}:2 &  \$13 \\ \hline
		$ T_{8} $ &	     \textit{d}:3  &  \$6 \\ \hline
		$ T_{9} $ &  	\textit{c}:5, \textit{d}:2, \textit{e}:5, \textit{b}:3   &  \$74 \\ \hline
		$ T_{10} $	&	\textit{e}:5, \textit{b}:3  &  \$65 \\ \hline	
	\end{tabular}
\end{table}

\begin{definition}
	\label{def_11}
	\rm According to the order of $\prec$, there may still be some items after pattern $X$ in transaction $T_q$ whose proportion in the transaction is defined as the remaining utility-occupancy (\textit{ruo}) \cite{gan2019huopm}, the formula of which is expressed as follows:
	\begin{equation}
	ruo(X,T_q) = \dfrac{\sum_{ i_j \notin X \wedge X \subseteq T_q \wedge X \prec i_j }u(i_j,T_q)}{tu(T_q)}.
	\end{equation}

	Furthermore, the remaining utility-occupancy of $X$ in a database is defined as follows \cite{gan2019huopm}:
	\begin{equation}
		ruo(X) = \dfrac{\sum_{X \subseteq T_q \wedge T_q \in \mathcal{D}}ruo(X,T_q)}{|\varGamma_X|)}.
	\end{equation}
\end{definition}

For example, as shown in TABLE \ref{table:db}, $ruo(a, T_1)$ = $(u(d, T_1)$ + $u(e, T_1)$ + $u(b, T_1))$/$tu(T_1)$ $ \approx $  0.8254. In addition,  $ruo(a)$ = $(ruo(a, T_1)$ + $ruo(a, T_2)$ + $ruo(a, T_3)$ + $ruo(a, T_5)$ + $ruo(a, T_6)$  + $ruo(a, T_7))$/6 = (0.8254 + 0.6452 + 0.5714 + 0.5714 + 0.8 + 0.7692)/6 = 0.6971. To facilitate the representation and mining of flexible HUOPs, we denote the required maximum length as \textit{maxlen} and the minimum length as \textit{minlen}.

\begin{definition}
	\label{def_14}
	\rm \textit{(Largest utility-occupancy in a transaction with respect to a pattern)}. Let a pattern $X$ exist in a transaction $T_q$ and the maximum length of patterns is set as \textit{maxlen}, and then put all items appearing after $X$ into a list in order of $\prec$ and denote the list as $V(X, T_q)$ = \{$i_1$, $i_2$, $\ldots$, $i_l$\}. Next, calculate the utility-occupancy corresponding to these items and express it as $W(X, T_q)$ = \{$uo(i_1, T_q)$, $uo(i_2, T_q)$, $\ldots$, $uo(i_l, T_q)$\}. The maximum number of items appended to $X$ is calculated as \textit{maxExtendLen} = \textit{maxlen} - $|X|$, where $|X|$ is the length of $X$. Thus, the largest utility-occupancy in transaction $T_q$ in regard to a pattern $X$ is the collection of the \textit{maxExtendLen} largest values in $W(X, T_q)$. For simplicity, we take this as $luo(X, T_q)$.
\end{definition}

For example, consider $a$ in $T_6$ in TABLE \ref{table:db}. Let \textit{maxlen} be 3; the utility-occupancy values of all items after $a$ in transaction $T_1$ can be calculated as \{0.1333, 0.5. 0.1667\}. Thus, \textit{maxExtendLen} = 3 - 1 = 2 and $luo(a, T_6)$ = \{0.5, 0.1667\}.

\begin{definition}
	\label{def_15}
	\rm \textit{(Revised remaining utility-occupancy)} Suppose there is a pattern $X$ in transaction $T_q$ and the maximum length of patterns is set as \textit{maxlen}. To reduce the upper bound on the utility-occupancy of $X$ in $T_q$ with a length constraint, we define the revised remaining utility-occupancy as $\textit{rruo}(X, T_q)$ = $\sum{luo(X, T_q)}$. Furthermore, the revised remaining utility-occupancy of a pattern $X$ in a transaction database $\mathcal{D}$ is calculated as $\textit{rruo}(X)$ =  $\dfrac{\sum_{X \in T_q, T_q \in \mathcal{D}}{\textit{rruo}(X, T_q)}}{|\varGamma_X|}$.
\end{definition}

Similar to the above example, we take a pattern $a$ in transaction $T_6$ as an example. The total utility-occupancy before a revision is 0.8. After optimization, however, the value reaches 0.6667, which is less than the original result. Moreover, in the entire database, the value of $\textit{rruo}(a)$ = 0.64315, which is much less than the former result of 0.6971.

\begin{property}
	\label{pro_1}
	\rm The upper bound of the revised remaining utility-occupancy of a pattern $X$ must be tighter than that of the remaining utility-occupancy.
\end{property}

\begin{proof}
	No matter which transaction $T_q$ pattern $X$ exists in, it is simple to obtain $\textit{rruo}(X, T_q)$ $\leq$ $\textit{ruo}(X, T_q)$. Hence, we can conclude that $\textit{rruo}(X)$ $ \leq $ $\textit{ruo}(X)$, the detailed proof of which is shown below:
	\begin{tabbing}	
		$rruo(X)$ \=
		=  $\dfrac{\sum_{X \in T_q, T_q \in \mathcal{D}}{\textit{rruo}(X, T_q)}}{|\varGamma_X|}$  \\
		\>$ \leq $  $\dfrac{\sum_{X \in T_q, T_q \in \mathcal{D}}{\textit{ruo}(X, T_q)}}{|\varGamma_X|}$ \\
		\> = $ruo(X)$.
	\end{tabbing}	
\end{proof}

We have found that the value of $\textit{rruo}(X)$ is smaller than that of $\textit{ruo}(X)$, and the next step is to determine how this property can be applied to find a tighter upper bound for the extension of pattern $X$. Before that, we should first build a data structure to store necessary information in the database.

\subsection{Revised List Structure for Storing Information}

In the previous sections, we introduced some basic concepts of the utility-occupancy and revised the remaining utility-occupancy. In this subsection, two compact data structures are designed for maintaining the essential information and avoiding going through databases multiple times. These structures are called a UO-nlist and FUO-table, respectively. In addition, as the main meaning of the nested list, there is one list within the larger list. The specific details of this are described below.
 
\begin{definition}
	\label{def_16}
	\rm (\textit{UO-nlist}) The UO-nlist related to a pattern $X$ is composed of several tuples, and one tuple corresponds to one transaction where $X$ occurs. Let $X$ be a pattern appearing in transaction $T_q$. We then define a tuple as consisting of three elements, namely, a transaction identifier $\textit{tid}$, the utility-occupancy of $X$ in $T_q$ (abbreviated as $\textit{uo}(X, T_q)$), and the largest utility-occupancy $\textit{luo}(X, T_q)$. As defined in Definition \ref{def_14}, $\textit{luo}(X, T_q)$ is a list recording a set of the largest occupancy utility of \textit{maxlen} - $|X|$ items after $X$ in $T_q$. Therefore, the UO-nlist can be marked as ($\textit{tid}$, $uo(X, T_q)$, $luo(X, T_q)$).
\end{definition}

For example, consider $ca$ in $T_6$, as shown in TABLE \ref{table:db}, and let \textit{maxlen} be 3. It is possible to obtain $uo(ca, T_6)$ = (\$6 + \$6) / \$60 = 0.2 and \textit{maxlen} - $|ca|$ = 3 - 2 = 1. Next, $luo(ca, T_6)$ = $\{uo(e, T_6)\}$ = $\{0.5\}$. Thus, UO-nlist of $ca$ in $T_6$ is \{$T_6$, 0.2, \{0.5\}\}. After scanning the database once, the UO-nlist of each item is constructed. For more details, refer to Fig. \ref{fig:UO-nlist}, which shows the support-ascending order of each item.

\begin{figure}[!htbp]
	\centering
	\includegraphics[scale=0.55]{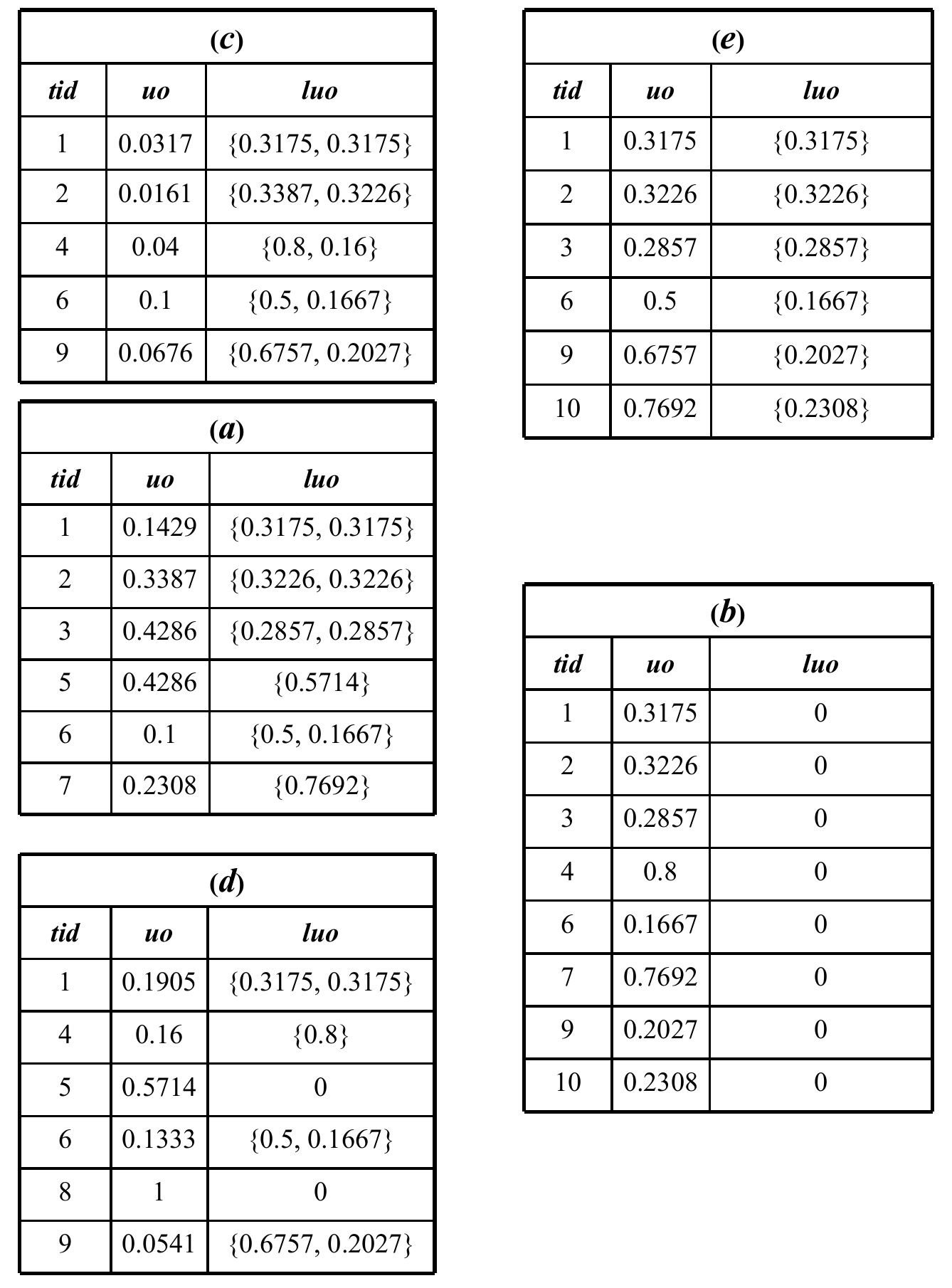}
	\caption{Constructed UO-nlists of each item in Table \ref{table:db}.}
	\label{fig:UO-nlist}
\end{figure}

Although UO-nlist already contains all necessary information for mining flexible HUOPs, it is troublesome to recalculate the elements on the list when we make use of the support, the utility-occupancy, and revised remaining utility-occupancy of a pattern. In this case, the execution time and other algorithm performance may be compromised. To remedy this problem, we further designed a data structure called a FUO-table, which is defined as follows:

\begin{definition}
	\label{def_17}
	\rm \textit{(Frequency-utility-occupancy table, FUO-table)} The FUO-table of a pattern $X$ consists of three elements, which are extracted from the corresponding UO-nlist. Among them, the support \textit{sup} is related to the number of tuples, $uo$ is the average utility-occupancy of $X$ on the UO-nlist, and \textit{rruo} is equal to the average of all values \textit{luo} defined in Definition \ref{def_15}. 
\end{definition}

To understand the concept of a FUO-table, take the construction process of the FUO-table of $c$ as an example, as displayed in Fig. \ref{fig:FWTableOfC}. Because $c$ in Fig. \ref{fig:UO-nlist} appears in five transactions, \textit{sup} is equal to 5. Here, $uo$ of $c$ in FUO-table is (0.0317 + 0.0161 + 0.04 + 0.1 + 0.0676)/5 = 0.05108. Then, the calculation process of \textit{rruo} is (0.3175 + 0.3175 + 0.3387 + 0.3226 + 0.8 + 0.16 + 0.5 + 0.1667 + 0.6757 + 0.2027)/5 = 0.76028. In addition, the FUO-tables of each item are shown in Fig. \ref{fig:FUO-table}.

\begin{figure}[!htbp]
	\centering
	\includegraphics[scale=0.55]{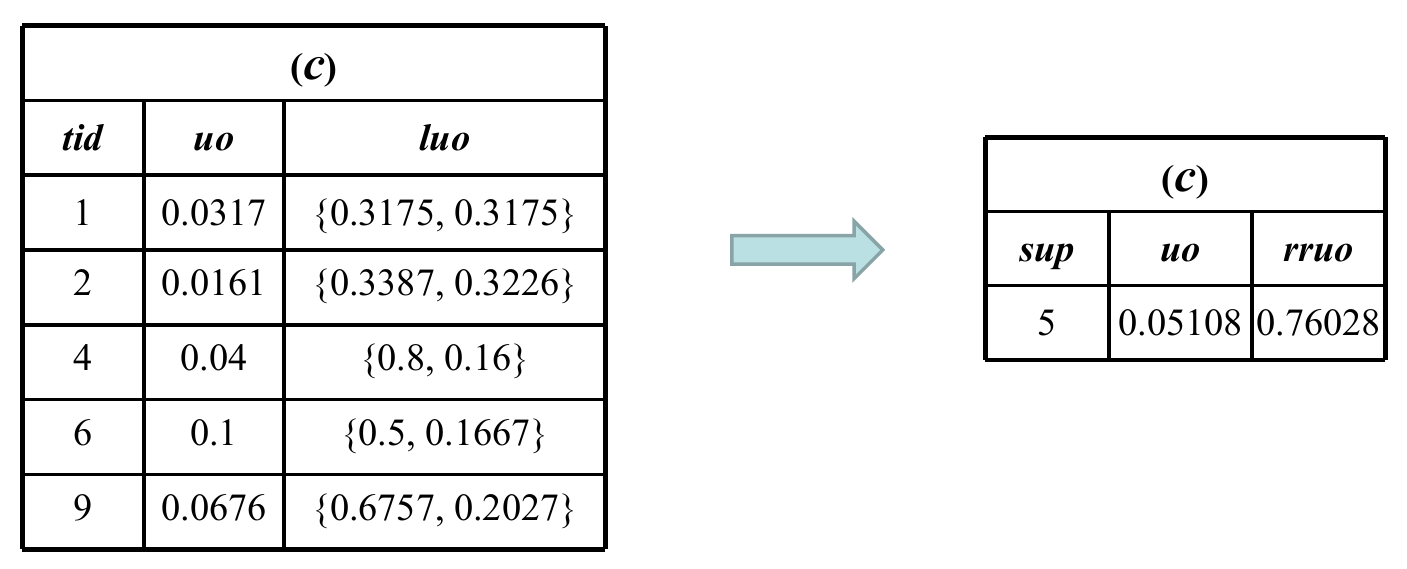}
	\caption{UO-nlist and FUO-table of item (c).}
	\label{fig:FWTableOfC}
\end{figure}

\begin{figure}[!htbp]
	\centering
	\includegraphics[scale=0.64]{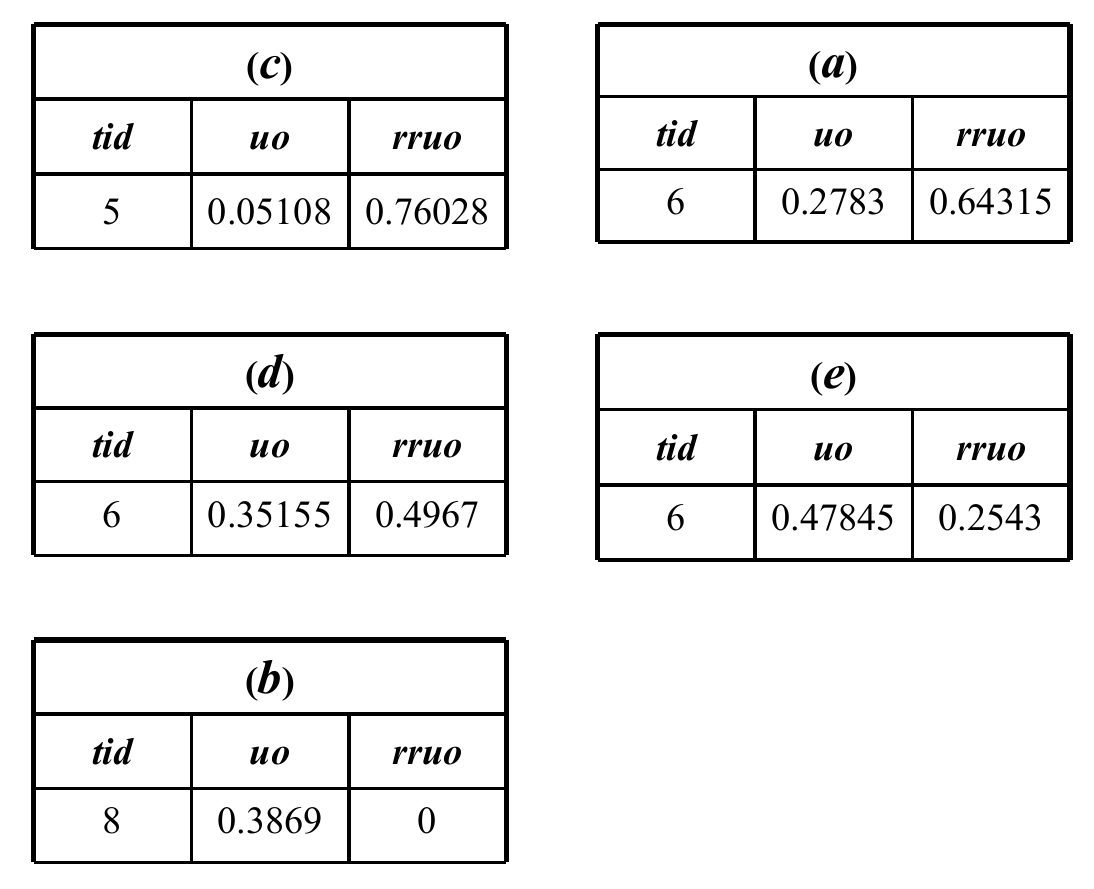}
	\caption{Constructed FUO-tables of all items.}
	\label{fig:FUO-table}
\end{figure}

Because promising or potential HUOPs are bound to meet the condition of frequent patterns, the initial UO-nlists and FUO-tables of frequent 1-itemsets are constructed after scanning the database twice. For the construction of the storage container of $k$-itemsets ($k > 1$), it is recommended to extend the already established subsets instead of scanning the database multiple times in a similar miner. Next, Algorithm 1 describes how to make use of subsets with the same prefix to calculate their extensions. To obtain the pattern $X_{ab}$, $X$ and its two extensions $X_a$ and $X_b$ ($a \prec b$) are applied in advance.  For simplicity, we address the UO-nlist of $ X_{ab}$ as \textit{X$_{ab}$.UONL} and FUO-table of $ X_{ab} $ as \textit{X$_{ab}$.FUOT}. Lines 2 and 3 and lines 18–22 in the algorithm check whether $X_a$ and $X_b$ appear in an entry. At the same time, each time an entry is checked, the support of $X_a$ is reduced by 1 until the value is less than $\alpha \times |D|$. The inherent reason for this is that when the support of the extensions of a pattern cannot satisfy the minimum support threshold, then it cannot be a HUOP, which is a necessary but insufficient condition.  Lines 4–17 illustrate the generation process of $X_{ab}$ when the construction conditions are met. If $X$ is an empty set, then the utility-occupancy of $X_{ab}$ is the total value of $X_a$ and $X_b$. Otherwise, the result is equal to the sum of the utility-occupancy of $X_a$ and $X_b$, and the utility-occupancy of $X$ is subtracted from $X$ because it is taken twice in this calculation. Moreover, the \textit{luo} of $X_{ab}$ is the same as $X_b$ owing to the latter order. In addition, the \textit{rruo} of $X_{ab}$ is the sum of the value in each $luo$, which is then divided by the support of $X_{ab}$. Thus far, the approach of building $(k+1)$-itemsets from $k$-itemsets has been realized.

\begin{algorithm}
	\label{Construction}
	\caption{Construction($ X $, $ X_{a} $, $ X_{b} $)}
	\begin{algorithmic}[1]
		\REQUIRE $X$, a pattern with its corresponding \textit{UO-nlist} and \textit{FUO-table}; $ X_{a} $, an extension of $X$ with an item $a$; $ X_{b} $, an extension of $X$ with an item $b$.
		\ENSURE	$ X_{ab} $.
		
		\STATE set \textit{X$_{ab}$.UONL} $\leftarrow$ $\emptyset $, $ X_{ab}.\textit{FUOT}$ $\leftarrow$ $\emptyset $;
		\STATE set \textit{supUB} = \textit{X$_{a}$.PFU.sup};
		\FOR {each tuple $ E_{a}$ $\in$ \textit{X$_{a}$.UONL} }	
		\IF {$ \exists E_{a}$ $\in X_{b}.UONL$ $\wedge$ $E_{a}.tid$ == $E_{b}.tid $}
		\IF{\textit{X.UONL} $\neq$ $\emptyset $}			
		\STATE search for $ E \in$ \textit{X.UONL}, $E.tid$ = $E_{a}.tid $;			
		
	    \STATE $E_{ab}$ $\leftarrow$ $<$$E_{a}.tid$, $E_{a}.uo$ + $E_{b}.uo$ - $E.uo$, \textit{E$_{b}$.luo}$>$;
	    \STATE \textit{$X_{ab}$.FUOT.uo} += $E_{a}.uo$ + $E_{b}.uo$ - $E.uo $;
		\ELSE
		\STATE $E_{ab}$ $\leftarrow$ $<$$E_{a}.tid$, $E_{a}.uo$ + $E_{b}.uo$, \textit{$E_{b}$.luo}$>$;
		\STATE \textit{$X_{ab}$.FUOT.uo}  += $ E_{a}.uo$ + $E_{b}.uo $;
		\ENDIF
		\FOR {each \textit{value} $\in$ \textit{E$_{ab}$.luo} }
		\STATE \textit{$ X_{ab}$.FUOT.rruo}  += \textit{value}; 
		\ENDFOR
		\STATE \textit{$X_{ab}$.UONL} $\leftarrow$  $X_{ab}.\textit{UONL}$ $\cup$ $E_{ab} $;
		\STATE \textit{X$_{ab}$.FUOT.sup} ++;	
		\ELSE		
		\STATE \textit{supUB} - -;
		\IF{\textit{supUB} $< \alpha \times |D| $}
		\STATE \textbf{return} \textit{null};	
		\ENDIF
		\ENDIF
		\ENDFOR
		\STATE $  X_{ab}.\textit{FUOT.uo}$ = $\dfrac{X_{ab}.\textit{FUOT.uo}}{X_{ab}.\textit{FUOT.sup}}$;
		\STATE $X_{ab}.\textit{FUOT.rruo}$ = $\dfrac{X_{ab}.\textit{FUOT.rruo}}{X_{ab}.FUOT.sup}$;
		\STATE \textbf{return} $ X_{ab} $
	\end{algorithmic}
\end{algorithm}

\subsection{Length-based Upper-Bound on Utility-Occupancy}

To the best of our knowledge, utility and utility-occupancy do not yet have downward closure properties such as the frequency, which results in a significant challenge with regard to increasing the performance of the algorithm and pruning the search space. Although Gan \textit{et al.} \cite{gan2019huopm} have previously designed an upper-bound on the utility-occupancy in the HUOPM algorithm, this upper-bound is relatively superficial for patterns with a length constraint. Thus, to achieve a tighter upper-bound, we propose the LUB strategy.

\begin{definition}
	\label{def_18}
	\rm \textit{(support count tree, $SC$-tree)}
	In this study, the support ascending order on items is applied to the complete algorithm and marked as $\prec$. For convenience, a set-enumeration tree called a $SC$-tree with an order of $\prec$ is built to simulate the depth-first search path.
\end{definition}

For a clearer description, refer to \cite{gan2019huopm}. For example, the extension nodes of $(ea)$ consist of (\textit{eab}), (\textit{ead}), and (\textit{eac}). If no action is taken, all possible nodes in a $SC$-tree are to be visited and their relative UO-nlists and FUO-tables are to be constructed, resulting in a massive space being occupied.  Suppose the existence of a pattern $X$ in transaction $T_q$, namely, $X \subseteq T_q$. After the database is revised in the order of $\prec$, all items that appear after $X$ in the transaction can be written as $T/X$. Inspired by the HUOPM algorithm \cite{gan2019huopm}, we propose the following lemmas.

\begin{lemma}
	\label{lemma_1}
	\rm Let $Y$ be an extension node of $X$. Here, $ \varGamma_X $ is a collection of transactions containing $X$, and $|\mathcal{D}|$ is the size of the given database. In addition, \textit{top} and $\downarrow$ imply the descending order of the sum of $uo$ and \textit{rruo} in each tuple and then take the top $\alpha$ $\times$ $|\mathcal{D}|$ as valid values for the numerator. The upper bound on the utility-occupancy of $Y$ can be calculated as follows:
	\begin{equation}
	\hat{\phi}(Y)) = \dfrac{\sum_{top \alpha \times |\mathcal{D}|, T_q \in \varGamma_X}\{uo(X,T_q) + rruo(X,T_q)\}^{\downarrow}}{|\alpha \times |\mathcal{D}||}.
	\end{equation}
\end{lemma}

\begin{proof}
    \label{pro_2}
    \rm Because $Y$ is obtained by appending items behind $X$, the equation $Y$/$X$ $\subseteq$ $T_q$/$X$ is derived. We can then achieve $ \sum{uo(Y / X, T_q)}$ $\leq$ $\sum{uo(T_q / X, T_q)} $. In addition, because the length of the required HUOPs cannot exceed \textit{maxlen}, the largest utility-occupancy  and revised remaining utility-occupancy in a transaction are applied to handle a tighter upper bound. For an easier formula derivation, we write $Y \subseteq T_q \wedge T_q \in \mathcal{D}$ as $Z$ and \textit{maxlen} as $M$.
	\begin{tabbing}	
		$ uo(Y)$ \=
		$= \dfrac{\sum_{Z \wedge |Y| \leq M}uo(Y,T_q)}{|\varGamma_Y|} $\\
		\>$ = \dfrac{\sum_{Z\wedge |Y/X| \leq M - |X|}(uo(X,T_q) + uo(Y/X,T_q))}{|\varGamma_Y|} $\\
	    \>$ \leq \dfrac{\sum_{Z \wedge |T_q/X| \leq M - |X|}(uo(X,T_q) + uo(T_q/X,T_q))}{|\varGamma_Y|} $\\
	
		\>$ \leq \dfrac{\sum_{Z}\{uo(X,T_q) + \sum_{v \in luo(X, T_q)}(v)\}}{|\varGamma_Y|} $\\ 
		\>$ \leq \dfrac{\sum_{Z}\{uo(X,T_q) + \sum_{v \in luo(X, T_q)}(v)\}}{|\varGamma_Y|} $\\ 
		\>$ = \dfrac{\sum_{Z}\{uo(X,T_q) + rruo(X, T_q)\}}{|\varGamma_Y|} $\\ 
		$\Longrightarrow  uo(Y) \leq	\dfrac{\sum_{Z}\{uo(X,T_q) + rruo(X, T_q)\}}{|\varGamma_Y|} $.
	\end{tabbing}

As the most critical step in the derivation process above, $uo(T_q/X$, $T_q)$ is less than the sum of all values in $luo(X$, $T_q)$ because the \textit{maxlen} - $|X|$ values in $luo(X$, $T_q)$ are definitely greater than the other utility-occupancy of items behind $X$.

We place all supporting transactions of $Y$ extending $X$ into $\varGamma_Y$; however, the value of $|\varGamma_Y|$ cannot be obtained until the entire UO-nlist of $Y$ has been constructed. Therefore, it is meaningless and contradictory to consider the above formula as the upper bound of the utility-occupancy of $Y$. Our purpose here is to minimize the construction of the list structures; however, this formula requires constructing a completed UO-nlist. In addition, when judging the high utility-occupancy pattern, we must first determine whether it is a frequent pattern, and if so, continue to judge whether the utility-occupancy meets the requirements. In addition, inequalities $\alpha \times |\mathcal{D}| \leq SC(Y) \leq SC(X)$ hold. Therefore, according to the discussion above, it is appropriate to replace the sum of $uo$ and \textit{rruo} in complete supporting transactions with the sum in the top $\alpha \times |\mathcal{D}|$ supporting transactions. The above formula can be further transformed to a tighter upper-bound by using the following property.

    \begin{tabbing}	
    	$uo(Y) \leq	\dfrac{\sum_{Z}\{uo(X,T_q) + \textit{rruo}(X, T_q)\}}{|\varGamma_Y|} $.\\
    	$  uo(Y) \leq	\dfrac{\sum_{top \alpha \times |\mathcal{D}|, T_q \in \varGamma_X}\{uo(X,T_q) + \textit{rruo}(X,T_q)\}^{\downarrow}}{|\alpha \times |\mathcal{D}||} $\\
    	$\Longrightarrow  uo(Y) \leq	\hat{\phi}(Y) $.
	\end{tabbing}
\end{proof}

Thus, given a pattern $X$, we can calculate the upper bounds $\hat{\phi}(Y)$ regarding the utility-occupancy with the length constraints of the nodes of all subtrees rooted at $X$.

\subsection{Pruning Strategies and Proposed Algorithm }

This section states the pruning strategies designed to prune the search space and improve the performance of the algorithm. Moreover, the proposed algorithm is also described in detail.

\begin{strategy}
	\label{stra_1}
	If the support count of a pattern $X$ in the designed $SC$-tree is greater than the minimum support threshold $\alpha$ multiplied by the size of the database $|\mathcal{D}|$, then this node is a frequent pattern; otherwise, this node and its subtree nodes rooted at the node can be directly pruned.
\end{strategy}

\begin{proof}
	The source of the above strategy is the Apriori \cite{agrawal1994fast} algorithm, the principle of which can generally be expressed as $SC(X_{k+1}) \leq SC(X_{k})$. Provided that $SC(X_{k})$ $\textless \alpha$ $\times$ $|\mathcal{D}|$, $SC(X_{k+1})$ $\textless \alpha$ $\times$ $|\mathcal{D}|$ is then acquired. Thus, node $X_{k}$ and the subtree rooted at this node can be trimmed immediately.
\end{proof}

\begin{strategy}
	\label{stra_2}
	In $SC$-tree, the upper bound of the child node $Y$ of node $X$ can be calculated immediately when the UO-nlist and FUO-table corresponding to $X$ have been constructed. If the derived upper bound is less than the predefined minimum utility-occupancy threshold $\beta$, then any nodes in the subtree of $X$ should be pruned.
\end{strategy}

\begin{proof}
	We concluded that the upper bound of $Y$ is certainly less than the real utility-occupancy of $Y$ in Lemma \ref{lemma_1}. Thus, if the upper bound is less than the minimum utility-occupancy threshold $\beta$, then $Y$ is not a HUOP. 
\end{proof}

\begin{strategy}
	\label{stra_3}
    Similar to the downward closure property in Strategy \ref{stra_1}, each time a tuple in $X_a$ is scanned, the support of $X_a$ decreases by one until the support is less than $\alpha \times |\mathcal{D}|$.
\end{strategy}

\begin{proof}
	The proof is the same as that of strategy \ref{stra_1}, only further limiting the support. 
\end{proof}

\begin{algorithm}
	\label{HUOPM$^+$-algorithm}
	\caption{HUOPM$^{+}$($\mathcal{D}$, \textit{ptable}, $\alpha$, $\beta$, \textit{minlen}, \textit{maxlen})}
	\begin{algorithmic}[1]		
		 \REQUIRE a transaction database $\mathcal{D}$, utility table \textit{ptable}, the minimum support threshold $ \alpha $, the minimum utility-occupancy threshold $\beta$, the minimum length \textit{minlen}, and the maximum length \textit{maxlen}.
		 \ENSURE  \textit{HUOPs}.
		 
		\STATE scan $\mathcal{D}$ to calculate the $SC(i)$ of each item $ i \in I $ and the $tu$ value of each transaction;
		\STATE find $ I^* \gets \left\{  i \in I | SC(i) \geq \alpha \times |\mathcal{D}| \right\} $, with respect to $ FP^1 $;
		\STATE sort $ I^* $ in the designed total order $ \prec $;
		\STATE using the total order $ \prec $, scan $\mathcal{D}$ once to build the UO-nlist and FUO-table for each 1-item $ i \in I^*$;
		\IF{\textit{maxlen} $\geq$ 1}
		\STATE \textbf{call \textit{HUOP$^+$-Search}}($\phi$, $I^*$, $\alpha$, $\beta$, \textit{minlen}, \textit{maxlen}).		
		\ENDIF
		\STATE \textbf{return} \textit{HUOPs}		
	\end{algorithmic}
\end{algorithm}

We introduced three feasible pruning strategies above. Next, the overall description concerning the proposed algorithm is presented as follows.

\begin{algorithm}
	\label{HUOP$^+$-Search procedure}
	\caption{HUOP$^+$-Search($X$, \textit{extenOfX}, $\alpha$, $\beta$, \textit{minlen}, \textit{maxlen})}
	\begin{algorithmic}[1]			
		\FOR {each itemset $ X_{a}\in $ \textit{extenOfX}}	
		\STATE obtain $ SC(X_a) $ and $ uo(X_a) $ from the $ X_{a}.\textit{FUOT} $;
		\IF{$ SC(X_a)$ $\geq \alpha$ $ \times |\mathcal{D}| $}
		
		\IF{$ uo(X_a)$ $\geq \beta \wedge$ $|X_a|\geq minlen$ }			 
		\STATE  \textit{HUOPs} $\leftarrow$ \textit{HUOPs} $\cup$ $X_{a} $;	
		\ENDIF
		
		\STATE	$ \hat{\phi}(X_a) \leftarrow $ \textbf{\textit{LengthUpperBound}}($X_a.UONL, \alpha) $;
		\IF{$ \hat{\phi}(X_a) \geq \beta $}
		
		\STATE \textit{extenOfX}$_{a}\leftarrow  \emptyset $;
		\FOR {each $ X_{b}$ $\in$ \textit{extenOfX} that $ X_{a} $   $ \prec $  $ X_{b} $}
		
		\STATE $ X_{ab}$ $\leftarrow$ $X_{a} \cup X_{b} $;
		
		\STATE call \textbf{\textit{Construction}}($X, X_{a}, X_{b}) $;
		\IF{$ X_{ab}.\textit{UONL}$ $\not= $ $\emptyset$ $\wedge SC(X_{ab})$ $\geq \alpha$ $\times |\mathcal{D}| $}
		
		\STATE \textit{extenOfX}$_{a} \leftarrow$ $\textit{extenOfX}_{a} \cup X_{ab} $;		
		\ENDIF			
		\ENDFOR	
		\IF {$|X| + 2 \leq maxlen$}	 						  		
		\STATE \textbf{call \textit{HUOP$^+$-Search}}$\boldmath{(X_{a}, \textit{extenOfX}_{a}, \alpha, \beta)}$;	
		\ENDIF
		\ENDIF
		\ENDIF
		\ENDFOR
		\STATE \textbf{return} \textit{HUOPs}
	\end{algorithmic}
\end{algorithm}

\begin{algorithm}
	\label{LengthUpperBound procedure}
	\caption{LengthUpperBound(\textit{X$_q$.UONL}, $ \alpha $)}
	\begin{algorithmic}[1]				
		\STATE \textit{sumTopK} $\leftarrow 0$, $\hat{\phi}(X_a)$ $\leftarrow 0$, $V_{occu}$ $\leftarrow$ $\emptyset $;
		\STATE calculate $ (uo(X,T_q)$ + \textit{rruo}($X,T_q$)) of each tuple from the built \textit{X$_{a}$.UONL}  and put them into the set of $ V_{occu} $;
		\STATE sort $ V_{occu} $ by descending order as $ V_{occu}^{\downarrow} $;
		
		\FOR {$ k \leftarrow $ 1 to  $\alpha \times |\mathcal{D}| $ in $ V_{occu}^{\downarrow} $}		
		\STATE \textit{sumTopK} $ \leftarrow$  \textit{sumTopK} +  $V_{occu}^{\downarrow}[k] $;					
		\ENDFOR
		
		\STATE $ \hat{\phi}(X_a)$ = $\dfrac{\textit{sumTopK}}{\alpha \times |\mathcal{D}|} $.	
		
		\STATE \textbf{return} $ \hat{\phi}(X_a) $
	\end{algorithmic}
\end{algorithm}

The main objective of the HUOPM$^+$ algorithm is to mine flexible high utility-occupancy patterns whose length is within a certain range. The pseudocode for the entire algorithm is shown in Algorithm 2. First, six parameters need to be entered in advance as the input of the algorithm, i.e., a transaction database $\mathcal{D}$, a profit table abbreviated as \textit{ptable}, the minimum support threshold $\alpha$ (0 $< \alpha \leq 1)$, the minimum utility-occupancy threshold $\beta$ (0 $< \beta \leq 1)$, the minimum length of the patterns \textit{minlen}, and the maximum length of the patterns \textit{maxlen} $(0 \leq $ \textit{minlen} $\leq $ \textit{maxlen}). This is the first time to scan the database to calculate the support of each item and the utility of each transaction. Then, rearrange the items in the transactions in support descending order ($\prec$) and thus obtain a revised database. Next, a rescanning of the revised database and establishment of the initial UO-nlist and FUO-table are necessary, which are the bases for the following search process. Finally, if the required maximum length \textit{maxlen} is greater than 1, the search procedure is executed.

To allow the search process to proceed in Algorithm 3, we need to input a prefix pattern $X$, a collection of extensions of $X$ that are initially distinct items, i.e., \textit{extenOfX}, $\alpha$, $\beta$, \textit{minlen}, and \textit{maxlen}. If a pattern $X_a$ in \textit{extenOfX} is expected to be output as a HUOP, it must satisfy three conditions. The support and utility-occupancy of $X_a$ are no less than the threshold $\alpha$ and $\beta$, respectively. In addition, it is still necessary to detect whether the length of $X_a$ is greater than or equal to \textit{minlen}. If the support of $X_a$ does not meet the restriction, then it is directly skipped, and the next pattern in the sequence is checked. Next, the upper bound in Algorithm 4 is calculated to determine whether the extending nodes of $X_a$ should be executed during a later calculation. If the upper bound $ \hat{\phi}(X_a) $ is no less than $\beta$, we can establish the extensions of $X_a$ and their relative UO-nlists and FUO-tables through the construction procedure described in Algorithm 1. For example, $X_{ab}$ is developed by merging $X_a$ with $X_b$ and using the specific process described above. Next, if the support of $X_{ab}$ meets the given requirement, it is placed in \textit{extenOfX}$_a$ for the next cycle. Finally, when the length of the extensions is less than \textit{maxlen}, the above mining program is recursed; otherwise, the program exits.

\section{Experiments} 
\label{sec:experiments}

To evaluate the performance of our proposed HUOPM$^+$ algorithm, some experiments conducted for comparison with the state-of-the-art utility-occupancy-based HUOPM algorithm are described in this section. 

\textit{Experimental environment}. First, an introduction of the computer parameters used will be quite helpful to allow readers to reproduce the experiments. The computer is running on Windows 10 with 7.88 GB of free RAM. We compared the presented algorithm HUOPM$^+$ with the state-of-the-art HUOPM algorithm for discovering the HUOPs. Both algorithms are implemented using Java.

\textit{Parameter settings}. As an innovation, the algorithm is mainly aimed at controlling the length of the output patterns; thus, to test this advantage, we set the maximum length of the patterns from one to five during the comparison between the HUOPM and HUOPM$^+$ algorithm outputs for all desired patterns. The minimum length is not a variable and is set to a constant of 1 because, even if the minimum length is greater than 1 and the maximum length is unchanged, the efficiency and performance of the algorithm are basically unchanged. Setting the minimum length ensures that our proposed algorithm can output patterns that are greater than the minimum length. The HUOPM$^+$ algorithm mainly involves two parameters, support and utility-occupancy. The following experiments analyze the runtime, visited nodes, and found patterns at different maximum lengths. For simplicity, we record the minimum support threshold as \textit{minsup} and the minimum utility-occupancy as \textit{minuo}. In the legend, the maximum length \textit{maxlen} in the HUOPM$^+$ algorithm is recorded as HUOPM$^{+}$-\textit{maxlen}. For example, HUOPM$^{+}$5 implies that the maximum length of the discovered patterns is 5.

\subsection{Tested Datasets}

Four standard datasets are applied to test the efficiency of the compared algorithms in terms of runtime, visited nodes, ie memory consumption, and patterns found. Three real-life datasets and one synthetic dataset were used, namely, retail, mushroom, kosarak, and T40I10D100K, respectively. It is worth noting that these datasets have different characteristics and that all aspects of the two algorithms (HUOPM and HUOPM$^{+}$) can be compared. We adopted a simulation method that is widely used in previous studies \cite{gan2019huopm,liu2012mining} to generate the quantitative and unit utility information for each object/item in each dataset. Next, we introduce these datasets in detail.

\begin{itemize}
	\item \textbf{retail\footnotemark[1]}: There are 88,162 transactions and 16,470 distinct items in the retail dataset, and the average transaction length is 76. Because there are few transactions and the average length of a transaction is not large, retail is a sparse dataset.

	\item \textbf{mushroom\footnotemark[1]}: This is a dense dataset, which has 8,124 transactions and 119 distinct items. 

	\item \textbf{kosarak\footnotemark[1]}: This has 990,002 transactions and 41,270 items, and note that the average length of its transactions is relatively longer, reaching up to 2,498, making it a dense dataset. 

	\item \textbf{T40I10D100K\footnotemark[2]}: This is a synthetic dataset with 942 distinct items and 100,000 transactions.
\end{itemize}

\footnotetext[1]{\url{http://www.philippe-fournier-viger.com/spmf/index.php?link=datasets.php}}
\footnotetext[2]{\url{http://fimi.uantwerpen.be/data/}}

\subsection{Runtime Analysis}

Let us first analyze the runtime of the compared algorithms. Figs. \ref{fig:runtimeMS} and \ref{fig:runtimeMU} show the runtime changes when the minimum support threshold $minsup$ and the minimum utility-occupancy threshold \textit{minuo} are gradually increasing. HUOPM$^{+}$1 to HUOPM$^{+}$5 mean that the maximum length of the patterns varies from 1 to 5. In general, it can be observed that, as one of the three thresholds increase, the execution time of each algorithm becomes increasingly short. In addition, when the maximum length of the discovered patterns is set to a distinct integer, the compared runtimes are clearly different. The runtime of the HUOPM algorithm can reach up to 2- to 5-times the runtime of HUOPM$^+$, which is crucial to the efficiency of the algorithms. As the reason for this phenomenon, the HUOPM$^+$ algorithm tightens the upper bound of the derived patterns, reduces the number of traversal nodes, and directly finds the derived patterns. From the comparison of the subplots in Figs. \ref{fig:runtimeMS} and \ref{fig:runtimeMU}, the performance of the proposed algorithm in compact datasets is better than that in sparse datasets, which is probably in that there are more longer patterns in compact datasets and plenty of patterns are pruned as a result of the designed utility-occupancy upper bound.

\begin{figure}[!hbt]
	\centering 
	\includegraphics[height=0.3\textheight,width=1.03 \linewidth,trim=80 0 50 0,clip,scale=0.35]{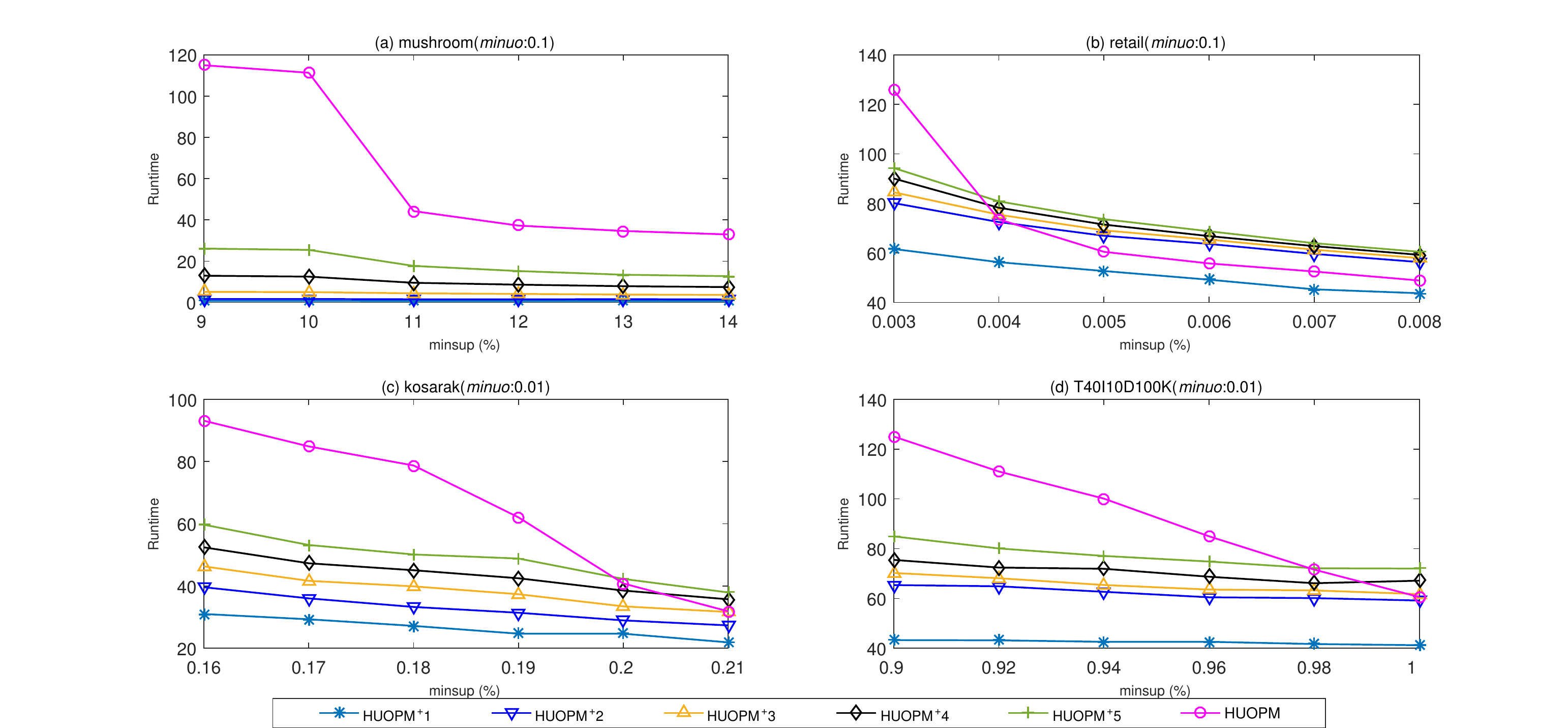}
	\caption{Runtime under a changed $minsup$ with a fixed $minuo$}
	\label{fig:runtimeMS}	
\end{figure}

\begin{figure}[!hbt]
	\centering 
	\includegraphics[height=0.3\textheight,width=1.03 \linewidth,trim=80 0 50 0,clip,scale=0.35]{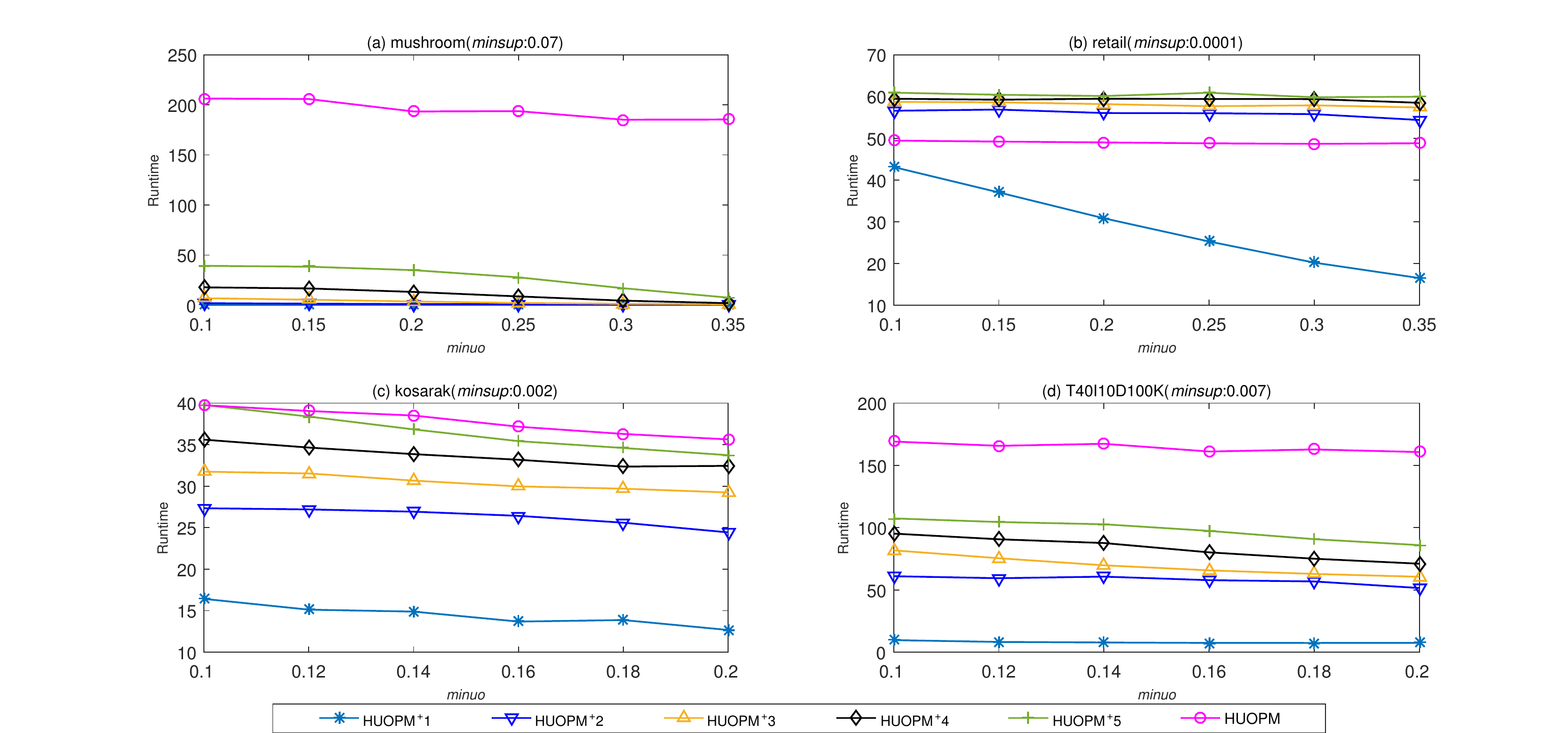}
	\caption{Runtime under a changed $minuo$ with a fixed $minsup$}
	\label{fig:runtimeMU}	
\end{figure}

\subsection{Visited Node Analysis}

In this subsection, we take into account the number of visited patterns, that is, we discuss the memory consumption. It is a common knowledge that, owing to the limitations of the current storage technology, memory consumption still accounts for a large proportion of the algorithm performance. Therefore, the fewer nodes visited during the process, the better the memory consumption of the algorithm.

\begin{figure}[!hbt]
	\centering 
	\includegraphics[height=0.3\textheight,width=1.05\linewidth,trim=80 0 50 0,clip,scale=0.35]{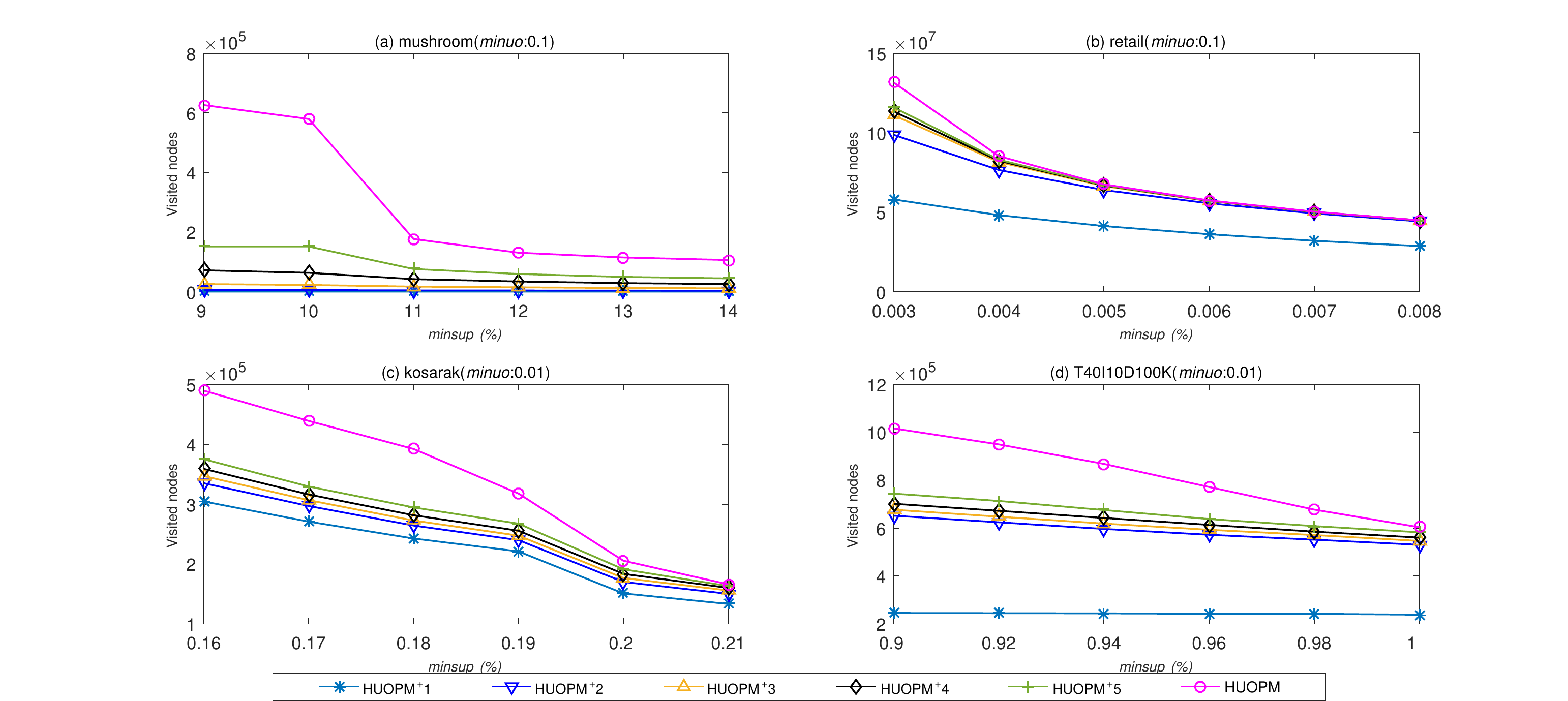}
	\caption{Visited nodes under a changed $minsup$ with a fixed $minuo$}
	\label{fig:memoryMS}	
\end{figure}

\begin{figure}[!hbt]
	\centering 
	\includegraphics[height=0.3\textheight,width=1.03\linewidth,trim=80 0 50 0,clip,scale=0.35]{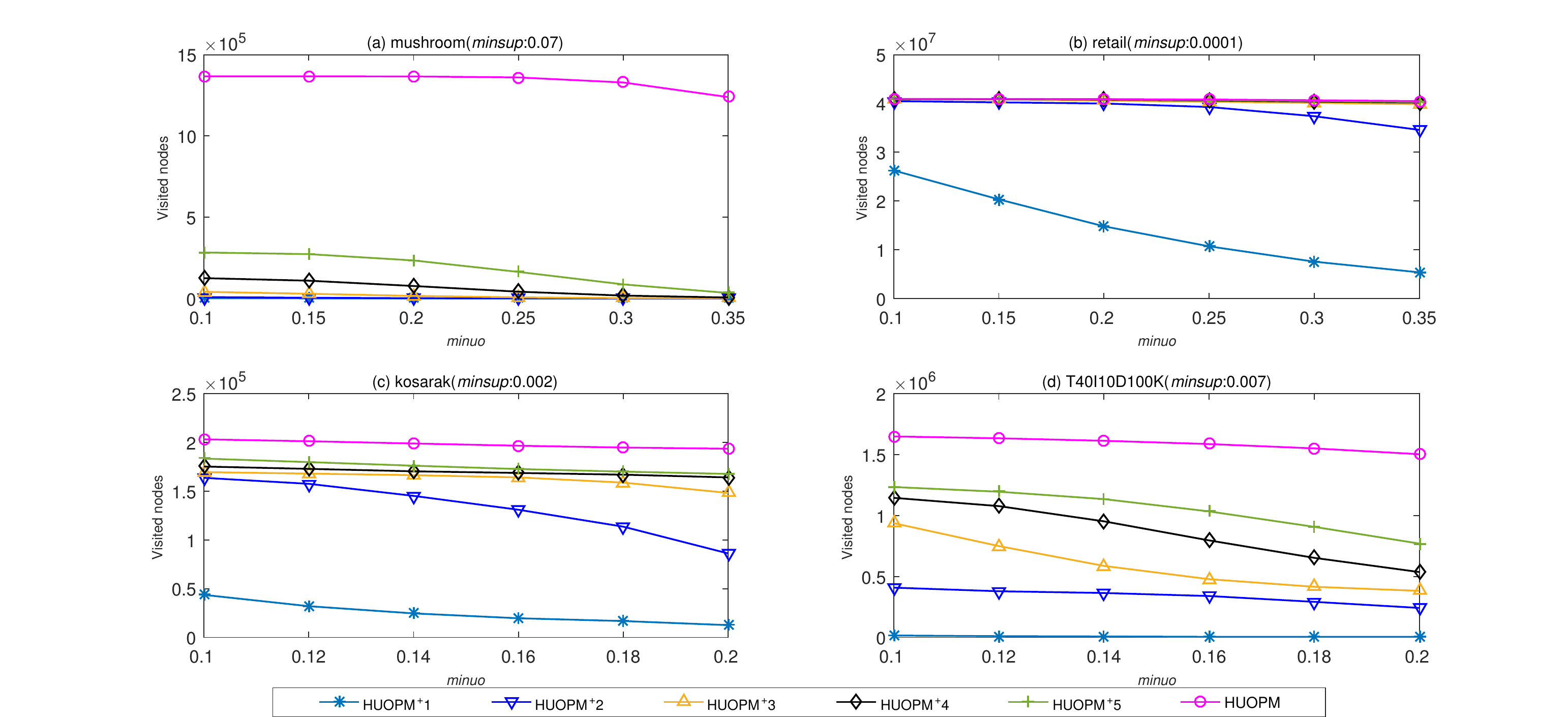}
	\caption{Visited nodes under a changed $minuo$ with a fixed $minsup$}
	\label{fig:memoryMU}	
\end{figure}

Figs. \ref{fig:memoryMS} and \ref{fig:memoryMU} respectively show that, when the support and utility-occupancy threshold change, the visited node changes. Clearly, the number of nodes visited by the HUOPM algorithm is greater than or equal to that of HUOPM$^+$. In particular, the smaller the maximum length set in HUOPM$^+$, the fewer the number of visited nodes. For example, in Fig. \ref{fig:memoryMS}, the minimum support of \textit{kosarak} is set to 0.0016, and the minimum utility-occupancy is set to 0.01. It can be seen that, when the maximum length is 5, the number of visited nodes of the HUOPM$^+$ algorithm is about one-third that of the HUOPM algorithm.

These phenomena suggest there is a difference between the actual memory consumption of the compared algorithms. It also illustrates the effectiveness of the mining model based on the pattern length-constraints proposed in this paper, that is, by setting different lengths, we not only can mine patterns that are more in line with the needs, we can also further reduce the actual memory consumption of the algorithm.

\subsection{Pattern Analysis}

The main goal of this study is the flexible mining of high utility-occupancy patterns, and during the process of mining to innovatively set the length-constraints. This allows patterns that are out of the range to not be traversed and the number of patterns found to be much less than that of HUOPM.

\begin{figure}[!hbt]
	\centering 
	\includegraphics[height=0.3\textheight,width=1.02\linewidth,trim=80 0 50 0,clip,scale=0.35]{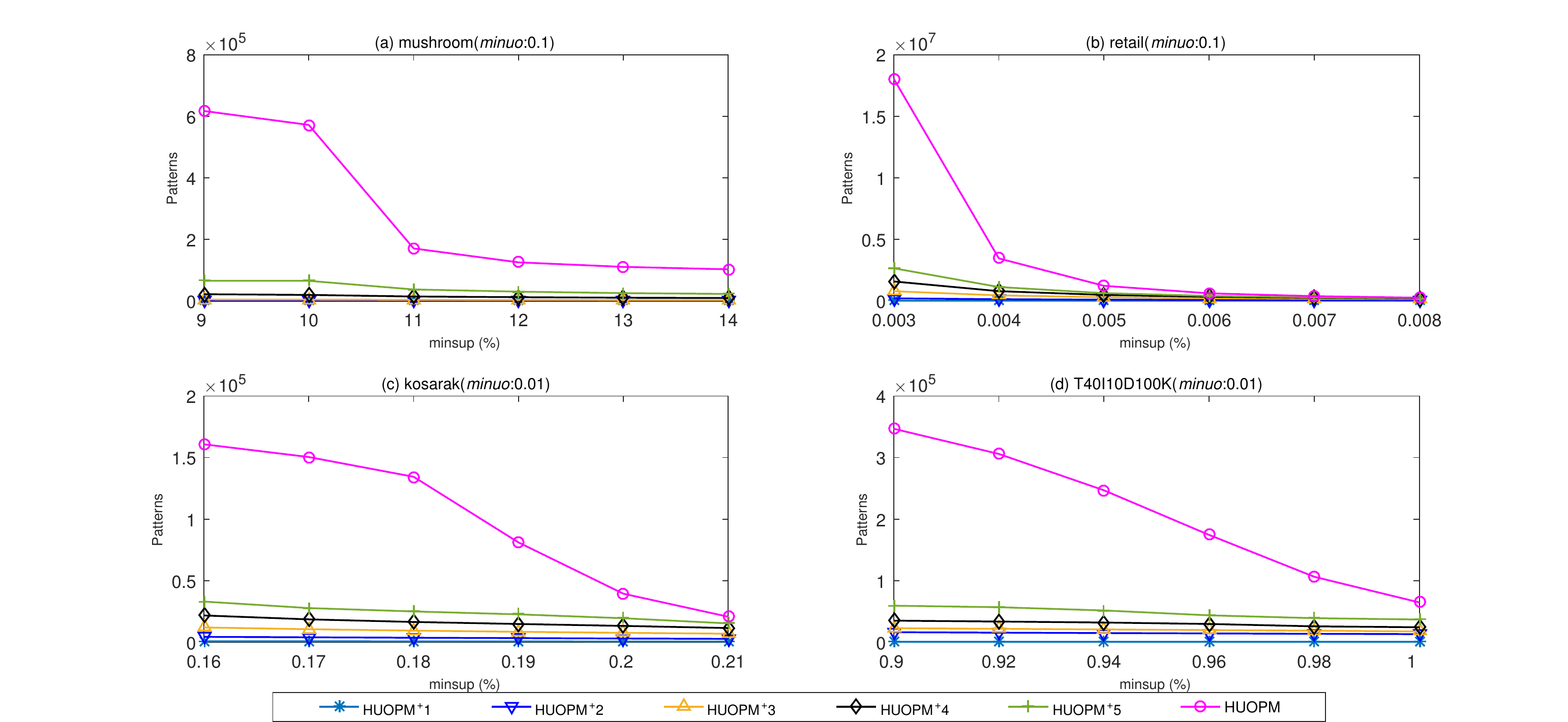}
	\caption{Patterns under a changed \textit{minsup} with a fixed \textit{minuo}}
	\label{fig:patternMS}	
\end{figure}

\begin{figure}[!hbt]
	\centering 
	\includegraphics[height=0.3\textheight,width=1.02\linewidth,trim=80 0 50 0,clip,scale=0.36]{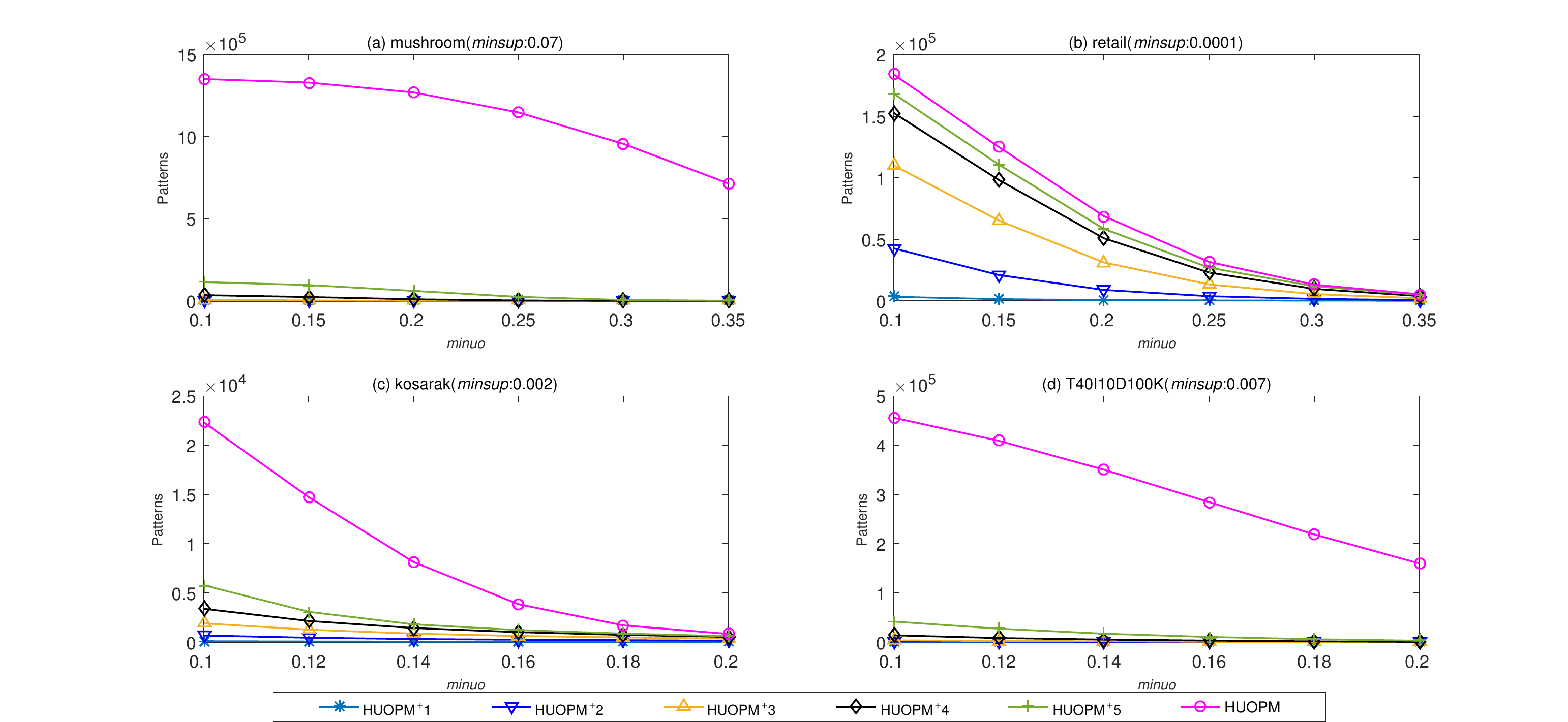}
	\caption{Patterns under a changed \textit{minuo} with a fixed \textit{minsup}}
	\label{fig:patternMU}	
\end{figure}

From the detailed results in Figs. \ref{fig:patternMS} and \ref{fig:patternMU}, we can observe that, in fact, the required patterns are often much fewer than the total number of patterns. Using this proposed algorithm, we not only can directly output the required patterns, we can also reduce the time required to process the data. For example, as shown in Fig. \ref{fig:patternMU} (b), $\alpha$ was set to 0.0001, and $\beta$ changed from 0.1 to 0.35 in increments of 0.05. The number of patterns desired from HUOPM$^+$1, HUOPM$^+$2, HUOPM$^+$3, HUOPM$^+$4, and HUOPM$^+$5 is 1,350, 20,946, 65,223, 98,166, and 110,676, respectively. In particular, the number of patterns of the mushroom dataset in Fig. \ref{fig:patternMU} with no length constraint is more than 10 times the number of patterns with \textit{maxlen} set to 5. From the two figures, we can notice that the number of flexibility patterns extracted by the proposed algorithm from sparse datasets is not significantly different from the number of the full patterns, while it is exactly the opposite in compact datasets. This is because the algorithm plays a stronger role in dense datasets.

\section{Conclusion and Future Studies} 
\label{sec:conclusion}

This paper proposes a novel algorithm called HUOPM$^+$, aimed at mining flexible high utility-occupancy patterns. It integrates length-constraints using a state-of-the-art HUOPM algorithm. In particular, unlike stopping the process of pattern growth directly, it deepens the length constraints of the procedure and narrows the upper bound by introducing the concept of a length upper-bound, which is one of the merits of the proposed algorithm. In addition, UO-nlist and FUO-table are designed to maintain the information in the database. The results of a subsequent experiment confirm that our proposed strategies can indeed discover HUOPs within a certain range of length, from \textit{minlen} to \textit{maxlen}, and greatly reduce the memory consumption and execution time. As part of a future study, we will apply the pattern length constraint to other utility mining algorithms, such as utility mining in dynamic profit databases \cite{nguyen2019mining}, utility-driven sequence mining \cite{gan2020proum} and nonoverlapping pattern mining \cite{wu2021ntp}.

\ifCLASSOPTIONcaptionsoff
  \newpage
\fi

\bibliographystyle{IEEEtran}
\bibliography{HUOPMPlus}

\end{document}